\documentclass[11pt]{article}

\usepackage{fullpage}
\usepackage[utf8]{inputenc}
\usepackage[intlimits]{amsmath}
\usepackage{amsfonts,amssymb,amsthm}
\usepackage{hyperref}
\usepackage{enumerate}
\usepackage{nicefrac}

\newtheorem{theorem}{Theorem}
\newtheorem{lemma}[theorem]{Lemma}
\theoremstyle{definition}

\DeclareMathOperator{\bs}{bs}
\DeclareMathOperator{\sens}{s}
\DeclareMathOperator{\fbs}{fbs}
\DeclareMathOperator{\orf}{OR}
\DeclareMathOperator{\Cert}{C}
\DeclareMathOperator{\fC}{FC}
\DeclareMathOperator{\RC}{RC}
\DeclareMathOperator{\Q}{Q}
\DeclareMathOperator{\D}{D}
\DeclareMathOperator{\R}{R}
\DeclareMathOperator{\M}{M}

\hypersetup{
	colorlinks = true,
	citecolor = blue
}

\begin{document}

\title{On Block Sensitivity and Fractional Block Sensitivity \thanks{
This work was supported by the European Union Seventh Framework Programme (FP7/2007-2013) under the QALGO (Grant Agreement No. 600700) project and the RAQUEL (Grant Agreement No. 323970) project, the ERC Advanced Grant MQC, and the Latvian State Research Programme NeXIT project No. 1.}}
\author{Andris Ambainis
 \and Krišjānis Prūsis
 \and Jevgēnijs Vihrovs}
\date{\small Centre for Quantum Computer Science, Faculty of Computing,\\ University of Latvia, Raiņa bulv.~19, Rīga, LV-1586, Latvia}

\maketitle

\begin{abstract}
We investigate the relation between the block sensitivity $\bs(f)$ and fractional block sensitivity $\fbs(f)$ complexity measures of Boolean functions.
While it is known that $\fbs(f) = O(\bs(f)^2)$, the best known separation achieves $\fbs(f) = \left(\nicefrac{1}{3\sqrt2} +o(1)\right) \bs(f)^{3/2}$.
We improve the constant factor and show a family of functions that give $\fbs(f) = \left(\nicefrac{1}{\sqrt6}-o(1)\right) \bs(f)^{3/2}.$
\end{abstract}

\section{Introduction}

The query complexity of Boolean functions is one of the simplest models of computation.
In this setting, the cost of the computation is the number of the input bits one needs to query to decide the value of the function on this input.
One of the main challenges is to precisely relate the computational power of the \emph{decision tree complexity} $\D(f)$, \emph{randomized decision tree complexity} $\R(f)$ and \emph{quantum decision tree complexity} $\Q(f)$ (see \cite{Aaronson_2016} for the currently known relations between various complexity measures).

\emph{Block sensitivity} $\bs(f)$ is a useful intermediate measure that has been used to show polynomial relations between the above measures.
\emph{Fractional block sensitivity} $\fbs(f)$ (aka fractional certificate complexity $\fC(f)$, randomized certificate complexity $\RC(f)$ \cite{Aaronson_2008}) is a recently introduced measure that is a relaxation of block sensitivity \cite{Tal_2013}.
It has been used to show a tight relation (up to logarithmic factors) between the \emph{zero-error randomized decision tree complexity} $\R_0(f)$ and \emph{two-sided bounded error randomized decision tree complexity} $\R_2(f)$ \cite{Kulkarni_2016}.

The relation between $\bs(f)$ and $\fbs(f)$ has been only partially understood.
On one hand, $\bs(f) \leq \fbs(f)$ and this inequality is tight.
On the other hand, it is known that $\fbs(f) \leq \bs(f)^2$ but the best known separation gives $\fbs(f) = \left(\nicefrac{1}{3\sqrt2}+o(1)\right) \bs(f)^{3/2}$ \cite{Gilmer_2016}.
We show a family of functions that give a constant factor improvement, $\fbs(f) = \left(\nicefrac{1}{\sqrt6}-o(1)\right) \bs(f)^{3/2}$.

\section{Definitions}

Let $f: \{0,1\}^n \rightarrow \{0,1\}$ be a Boolean function on $n$ variables.
We denote the input to $f$ by a binary string $x = (x_1, \ldots, x_n)$, so that the $i$-th variable is $x_i$.
For an index set $P \subseteq [n]$, let $x^P$ be the input obtained from an input $x$ by flipping every bit $x_i$, $i \in P$.

We briefly define the notions of sensitivity, certificate complexity and variations on them.
For more information on them and their relations to other
complexity measures (such as deterministic, probabilistic and quantum decision
tree complexities), we refer the reader to the surveys by Buhrman and de Wolf \cite{Buhrman_deWolf_2002}
and Hatami et al. \cite{Hatami_Kulkarni_Pankratov_2011}.

The \emph{sensitivity complexity} $\sens(f,x)$ of $f$ on an input $x$ is defined as \begin{equation} \sens(f,x) = | \{ i \in [n] \mid f(x) \neq f(x^{\{i\}})\} |.\end{equation} The \emph{sensitivity} $\sens(f)$ of $f$  is defined as $\max_{x \in \{0,1\}^n} \sens(f,x)$.

The \emph{block sensitivity} $\bs(f,x)$ of $f$ on an input $x$ is defined as the maximum number $t$ such that there are $t$ pairwise disjoint subsets $B_1, \ldots , B_t$ of $[n]$ for which $f(x) \neq f\left(x^{B_i}\right)$. We call each $B_i$ a \emph{block}.
The \emph{block sensitivity} $\bs(f)$ of $f$  is defined as $\max_{x \in \{0,1\}^n} \bs(f,x)$.

The \emph{fractional block sensitivity} $\fbs(f,x)$ of $f$ on an input $x$ is the optimal value of the following linear program, where each sensitive block of $x$ is assigned a real valued weight $w_B$:
\begin{align*}
\max \sum_{f(x) \neq f(x^B)} w_B \hspace{1.5cm} \text{subject to: } &\forall i \in [n]: \sum_{B \ni i} w_B \leq 1, \\
&\forall B: 0 \leq w_B \leq 1.
\end{align*}
The \emph{fractional block sensitivity} of $f$ is defined as $\fbs(f) = \max_{x \in \{0,1\}^n} \fbs(f,x)$.

A \emph{certificate} $C$ of $f$ is a partial assignment $C: P \rightarrow \{0,1\}, P \subseteq [n]$ of the input such that $f$ is constant on this restriction. We call $|P|$ the \emph{length} of $C$. If $f$ is always 0 on this restriction, the certificate is a \emph{0-certificate}. If $f$ is always 1, the certificate is a \emph{1-certificate}.

The \emph{certificate complexity} $\Cert(f,x)$ of $f$ on an input $x$ is defined as the minimum length of a certificate that $x$ satisfies.
The \emph{certificate complexity} $\Cert(f)$ of $f$  is defined as $\max_{x \in \{0,1\}^n} \Cert(f,x)$.

The \emph{fractional certificate complexity} $\fC(f,x)$ of $f$ on an input $x$ is the optimal value of the following linear program, where each position $i \in [n]$ is assigned a real valued weight $v_i$:
\begin{align*}
\min \sum_{i \in [n]} v_i \hspace{1.5cm} \text{subject to: } &\forall B \text{ s.t. } f(x) \neq f(x^B): \sum_{i \in B} v_i \geq 1, \\
&\forall i \in [n]: 0 \leq v_i \leq 1.
\end{align*}
The \emph{fractional certificate complexity} of $f$ is defined as $\fC(f) = \max_{x \in \{0,1\}^n} \fC(f,x)$.

For any of these measures $\M \in \{\sens, \bs, \fbs, \fC, \Cert\}$, define $\M_b(f) = \max_{x \in f^{-1}(b)} \M(f,x)$.
In that way, we define the measures $\sens_0(f)$, $\sens_1(f)$, $\bs_0(f)$, $\bs_1(f)$, $\fbs_0(f)$, $\fbs_1(f)$, $\fC_0(f)$, $\fC_1(f)$, $\Cert_0(f)$, $\Cert_1(f)$.
In particular, $\M(f) = \max\{\M_0(f), \M_1(f)\}$.

One can show that $\sens(f) \leq \bs(f) \leq \fbs(f) \leq \fC(f) \leq \Cert(f)$ \cite{Tal_2013}.
In fact, the linear programs of $\fbs(f,x)$ and $\fC(f,x)$ are duals of each other.
Therefore, $\fbs(f) = \fC(f)$.

\section{Separation}

The separation in \cite{Gilmer_2016} composes a graph property Boolean function (namely, whether a given graph is a star graph) with the $\orf$ function.
We build on these ideas and define a new graph property $g$ for the composition that gives a larger separation.

\begin{theorem}
There exists a family of Boolean functions such that $$\fbs(f) = \left(\frac{1}{\sqrt6}-o(1)\right) \bs(f)^{3/2}.$$
\end{theorem}

\begin{proof}

Let $N \geq 12$ be a multiple of 3.
An input on $\binom{N}{2}$ variables $(x_{1,2}, x_{1,3},\ldots,x_{N-1,N})$ encodes a graph $G$ on $N$ vertices.
Let $x_{i,j} = 1$ iff the vertices $i$ and $j$ are connected by an edge in $G$.

We define an auxiliary function $g : \{0,1\}^{\binom{N}{2}} \rightarrow \{0,1\}$.
Partition $[N]$ into three sets $S_0,S_1,S_2$ such that $S_r = \{i \in [N] \mid i \equiv r \pmod 3 \}$.
Let $g(x)=1$ iff:
\begin{itemize}
\item there is some vertex $i$ that is connected to every other vertex by an edge (a star graph);
\item for any $r \in \{0, 1, 2\}$, no two vertices $j, k \neq i$ such that $j, k \in S_r$ are connected by an edge.
\end{itemize}
Formally, $g(x)=1$ iff $x$ satisfies one of the following 1-certificates $C_1,\ldots,C_N$:
$C_i$ assigns 1 to every edge in $\{x_{j,k} \mid j = i \lor k = i \} $, and assigns 0 to every edge in $\{x_{jk} \mid j \neq i, k \neq i, j \equiv k \pmod 3 \}$.

Now we calculate the values of $\bs_0(g), \bs_1(g), \fbs_0(g)$.
\begin{itemize}
\item $\bs_0(g) = 3$.

Consider an input $x$ describing a triangle graph between vertices $i, j, k$.
For this input $g(x) = 0$.
Let $x'$ be an input obtained from $x$ by removing the edge $x_{i,j}$ and adding all the missing edges $x_{k,l}$, for all $l \neq i, j$.
The corresponding graph is a star graph, therefore, $g(x') = 1$.
Let $B_k$ be the sensitive block that flips $x$ to $x'$.
Similarly define $B_i$ and $B_j$.
None of the three blocks overlap, hence $\bs_0(g,x) \geq 3$.

Now we prove that $\bs_0(g) \leq 3$.
Assume the contrary, that there exists an input $x \in f^{-1}(0)$ with $\bs(g,x) \geq 4$.
Then $x$ has (at least) 4 non-overlapping sensitive blocks $B_1,\ldots, B_4$.
Each $x^{B_i}$ satisfies one of the 1-certificates, each a different one.
There are 4 such certificates, therefore at least two of them require a star at vertices $i, j$ belonging to the same $S_r$.
The corresponding certificates $C_i$ and $C_j$ both assign 1 at the edge $x_{i,j}$.
On the other hand, every other $C_k$ assigns 0 at $x_{i,j}$.
Therefore, of the 4 certificates corresponding to $B_1, \ldots, B_4$, two assign 1 to this edge and two assign 0 to this edge.
Then, regardless of the value of $x_{i,j}$, we would need to flip it in two of the blocks $B_1,\ldots, B_4$: a contradiction, since the blocks don't overlap.
Therefore no such $x$ exists.

\item $\bs_1(g) =\frac{ N^2}{6}+\frac{N}{6}$.

Examine any 1-certificate $C_i$.
Find three indices $j,k,l \equiv i \pmod 3$ (this is possible, as $N \geq 12$).
Any input $x$ that satisfies $C_i$ has $x_{i,j} = x_{i,k} = x_{i,l} = 1$.
On the other hand, any other 1-certificate $C_t$ requires at least two of the variables $x_{i,j}, x_{i,k}, x_{i,l}$ to be 0.
Hence, the Hamming distance between $C_i$ and $C_t$ is at least two.
Therefore, flipping any position of $x$ that is fixed in $C_i$ changes the value of the function as well.
Thus, $\sens(f,x) = \Cert(f,x)$.
As $\sens(f,x) \leq \bs(f,x) \leq \Cert(f,x)$, we have $$\bs(f,x) = \Cert(f,x) = |C_i| = 3{\binom{N/3}{2}} + \frac{2N}{3} = \frac{ N^2}{6}+\frac{N}{6}.$$

\item $\fbs_0(g) \geq \frac N 2$.

Examine the all zeros input $0^{\binom{N}{2}}$.
Any sensitive block $B$ of this input flips the edges on a star from some vertex.
Therefore, any position is flipped by exactly two of the sensitive blocks.
The weights $w_B = \frac 1 2$ for each sensitive block $B$ then give a feasible solution for the fractional block sensitivity linear program.
As there are $N$ sensitive blocks, $\fbs(g,0^{\binom{N}{2}}) = \frac N 2$.

\end{itemize}

To obtain the final function we use the following lemma:
\begin{lemma}[Proposition 31 in \cite{Gilmer_2016}]
Let $g$ be a non-constant Boolean function and $$f = \orf(g^{(1)}, \ldots ,g^{(m)}),$$ an $\orf$ composed with $m$ copies of $g$.
Then for complexity measures $\M \in \{\bs, \fbs\}$, we have
\begin{align*}
\M_1(f) &= \M_1(g) \\
\M_0(f) &= m \cdot \M_0(g).
\end{align*}
\end{lemma}
Let $m=\bs_1(g)/\bs_0(g)=\frac{N^2}{18}+\frac{N}{18}$.
Then $\bs(f) = \bs_0(f)=\bs_1(f) = \bs_1(g) = \frac{N^2}{6}+\frac{N}{6}$.
On the other hand, $\fbs(f) \geq \fbs_0(f)=m\cdot \fbs_0(g) \geq m \cdot \frac N 2 = \frac{N^3}{36}+\frac{N^2}{36}$.
Therefore, we have $$\fbs(f) \geq \left( \frac{N^2}{6}+\frac{N}{6} \right) \cdot \frac{N}{6} = \bs(f) \cdot \left(\frac{1}{\sqrt 6}-o(1)\right) \sqrt{\bs(f)} = \left(\frac{1}{\sqrt 6}-o(1)\right) \bs(f)^{3/2}.$$
\end{proof}

\bibliographystyle{alpha}
\bibliography{bibliography}

\end{document}